\documentclass{llncs}
\usepackage{algorithm}
\usepackage{algorithmic}
\usepackage{amsfonts}
\title{Improved Interference in Wireless Sensor Networks}
\author{Pragya Agrawal\inst{1} \and Gautam K. Das\inst{2}}
\institute{Department of Electronics and Communication Engineering \\ Visvesvaraya National Institute of Technology, Nagpur - 440011, India
 \and
Department of Mathematics, Indian Institute of Technology Guwahati \\ Guwahati - 781039, India}
\date{}

\begin{document}
\maketitle
\input{psfig.sty}

\begin{abstract}
Given a set ${\cal V}$ of $n$ sensor node distributed on a $2$-dimensional plane and a source
node $s \in {\cal V}$, the {\it interference problem} deals with assigning transmission range 
to each $v \in {\cal V}$ such that the members in ${\cal V}$ maintain connectivity predicate 
${\cal P}$, and the maximum/total interference is minimum. We propose algorithm for both 
{\it minimizing maximum interference} and {\it minimizing total interference} of the 
networks. For minimizing maximum interference we present optimum solution with running time 
$O(({\cal P}_n + n^2) \log n)$ for connectivity predicate ${\cal P}$ like strong connectivity, 
broadcast ($s$ is the source), $k$-edge(vertex) connectivity, spanner, where $O({\cal P}_n)$ 
is the time complexity for checking the connectivity predicate ${\cal P}$. The running time of 
the previous best known solution was $O({\cal P}_n \times n^2)$ \cite{BP08}.

For the minimizing total interference we propose optimum algorithm for the connectivity 
predicate broadcast. The running time of the propose algorithm is $O(n)$. For the same 
problem, the previous best known result was $2(1 + \ln (n-1))$-factor approximation 
algorithm \cite{BP08}. We also propose a heuristic for minimizing total interference in 
the case of strongly connected predicate and compare our result with the best result 
available in the literature. Experimental results demonstrate that our heuristic outperform 
existing result.

\end{abstract}

\section{Introduction}
A sensor node is a small size, low-power devices with limited computation and
communication capabilities. A wireless sensor network (WSN) is a set of sensor nodes 
and each sensor of the network can measure certain physical phenomena like temperature, 
pressure, intensity of light or vibrations around it. Wireless networks of such sensor 
nodes have many potential applications, such as surveillance, environment monitoring and 
biological detection \cite{ASSC02,MPSCA02}. The objective of such network is to process
some high-level sensing tasks and send the data to the application \cite{FKKLS07}.

Since the sensor nodes are battery operated, so minimizing energy consumption is a critical
issue for designing topology of wireless sensor networks to increase its lifetime.
A message packet transmitted by a sensor device to another is often received by many
nodes in the vicinity of the receiver node. This causes collision of signals and increased
interference in its neighboring nodes, which leads to message packet loss. Due to packet
loss the sender node needs to retransmit the message packets. Therefore, large interference
of networks may result delays for delivering data packets. This would further enhance the
energy consumption of nodes in network. Thus interference reduction of nodes is a crucial
issue for minimizing (i) delays for delivering data packets and (ii) energy consumption of
a wireless sensor network.

\section{Network Model and Interference Related Problems}
\begin{definition}
$\delta(u, v)$ denotes the Euclidean distance between two sensor nodes $u$ and $v$. 
A range assignment is $\rho:{\cal V} \rightarrow \mathbb{R}$. A communication graph 
$G_\rho = ({\cal V}, E_\rho)$ is a directed graph, where ${\cal V}$ is the set of 
sensor nodes and $E_\rho = \{(u, v) | \rho(u) \geq \delta(u, v)\}$. $G_\rho$ is 
said to be strongly connected if there is a directed path between each pair of 
vertices $u, v \in {\cal V}$ in $G_\rho$. $G_\rho$ is said to be an arborescence 
rooted at $v \in {\cal V}$ if there is a directed path from $v$ to all other 
vertices $u \in {\cal V}$ in $G_\rho$.
\end{definition}

Different models have been proposed to minimize the interference in sensor networks
\cite{BRWZ04,RWZ04,RSWZ05}. In this paper, we focus on the following two widely accepted models:

\begin{itemize}
\item[$\bullet$] Sender Interference Model (SIM): The interference of a node $u \in {\cal V}$ 
is the cardinality of the set of nodes to whom it can send messages directly. The set of nodes 
creating interference with $u$ with respect to the range assignment $\rho$ is denoted by 
$I_S^u(\rho)$ and defined by $I_S^u(\rho) = \{v \in {\cal V} \setminus \{u\} | \delta(u, v) 
\leq \rho(u)\}$, where $\rho(u)$ is the range of the node $u$. Therefore the interference 
value of node $u$ is equal to $|I_S^u(\rho)|$ with respect to range assignment $\rho$.

\item[$\bullet$] Receiver Interference Model (RIM): The interference of a node $u \in {\cal V}$ 
is the cardinality of the set of nodes from which it can receive messages directly. The set of 
nodes creating interference with $u$ with respect to the range assignment $\rho$ is denoted by 
$I_R^u(\rho)$ and defined by $I_R^u(\rho) = \{v \in {\cal V} \setminus \{u\} | \delta(u, v) \leq 
\rho(v)\}$, where $\rho(v)$ is the range of the node $v$. Therefore the interference value of 
node $u$ is equal to $|I_R^u(\rho)|$ with respect to range assignment $\rho$.
\end{itemize}

In the interference minimization problems, the objective is to minimize the following 
objective functions:

\begin{itemize}
\item[$\circ$] {\bf MinMax SI (MMSI):} Given a set ${\cal V}$ of $n$ sensors, find a range 
assignment function $\rho$ such that the communication graph $G_\rho$ satisfy the connectivity 
predicate ${\cal P}$ and maximum sender interference of the sensors in the network is minimum i.e., 
\[\min_{\rho|G_\rho {satisfy} {\cal P}} \max_{v \in {\cal V}} I_S^v(\rho)\]

\item[$\circ$] {\bf Min Total SI (MTSI):} Given a set ${\cal V}$ of $n$ sensors, find a range 
assignment function $\rho$ such that the communication graph $G_\rho$ satisfy the connectivity 
predicate ${\cal P}$ and total sender interference of the entire sensor network is minimum i.e., 
\[\min_{\rho|G_\rho {satisfy} {\cal P}} \sum_{v \in {\cal V}} I_S^v(\rho)\]

\item[$\circ$] {\bf MinMax RI (MMRI):} Given a set ${\cal V}$ of $n$ sensors, find a range
assignment function $\rho$ such that the communication graph $G_\rho$ satisfy the connectivity
predicate ${\cal P}$ and maximum receiver interference of the sensors in the network is minimum i.e.,
\[\min_{\rho|G_\rho {satisfy} {\cal P}} \max_{v \in {\cal V}} I_R^v(\rho)\]

\item[$\circ$] {\bf Min Total RI (MTRI):} Given a set ${\cal V}$ of $n$ sensors, find a range
assignment function $\rho$ such that the communication graph $G_\rho$ satisfy the connectivity
predicate ${\cal P}$ and total receiver interference of the entire sensor network is minimum i.e.,
\[\min_{\rho|G_\rho {satisfy} {\cal P}} \sum_{v \in {\cal V}} I_R^v(\rho)\]
\end{itemize}

In the communication graph corresponding to a range assignment $\rho$, the number of 
out-directed edges and in-directed edges are same, which leads to the following result:

\begin{theorem} \label{theorem-1}
For any range assignment $\rho$, {\bf MTSI = MTRI}.
\end{theorem}

\subsection{Our Contribution}
In this paper, we propose algorithm for {\it minimizing maximum interference} and 
{\it minimizing total interference} of a given wireless sensor network. For minimizing maximum
interference we present optimum solution with running time $O(({\cal P}_n + n^2) \log n)$ for
connectivity predicate ${\cal P}$ like strong connectivity, broadcast, $k$-edge(vertex)
connectivity, spanner. Here $O({\cal P}_n)$ is the time complexity for checking the 
connectivity predicate ${\cal P}$. The running time of the previous best known solution was
$O({\cal P}_n \times n^2)$ \cite{BP08}.

For the minimizing total interference we propose optimum algorithm for the connectivity predicate 
broadcast. The running time of the propose algorithm is $O(n)$. For the same problem, the previous
best known result was $2(1 + \ln (n-1))$-factor approximation algorithm \cite{BP08}. We also
propose a heuristic for minimizing total interference in the case of strongly connected predicate
and compare our result with the best result available in the literature. Experimental results
demonstrate that our heuristic outperform existing result.

We organize remaining part of this paper as follows: In Section \ref{RelatedWork}, we discuss
existing results in the literature. The algorithm for the optimum solution of minimizing 
maximum interference for different connectivity predicate appears in Section \ref{MinMax}.
In Section \ref{MinTotal}, we present optimum algorithm to minimize total interference for
the connectivity predicate broadcast. A heuristic for minimizing total interference for the 
strongly connected predicate appears in Section \ref{Heuristic}. Finally, we conclude the
paper in Section \ref{Conclusion}.

\section{Related Works} \label{RelatedWork}
Tan et al. \cite{TLLWC11} studied minimization of the average interference and the maximum
interference for the highway model, where all the nodes are arbitrarily distributed on a
line. For the minimum average interference problem they proposed an exact algorithm, which
runs in $O(n^3 \Delta^3)$ time, where $n$ is the number of nodes and $\Delta$ is the maximum
node degree in the communication graph for the equal range assigned to each nodes equal to
the maximum consecutive distance between two nodes. For the minimization of maximum interference
problem, they proposed $O(n^3 \Delta^{O(k)})$ time algorithm, where $k = O(\sqrt{\Delta})$.
Lou et al. \cite{LTWL12} improves the time complexity to $O(n \Delta^2)$ for the minimization
of average interference problem. Rickenbach et al. \cite{RSWZ05} proved that minimum value of  
maximum interference is bounded by $O(\sqrt{\Delta})$ and presented an $O(\sqrt[4]{\Delta})$-factor
approximation algorithm, where $\Delta$ is the maximum node degree in the communication graph
for some equal range $\rho_{\max}$ assigned to all the nodes.

For 2D networks, Buchin \cite{Buchin08} considered receiver interference model and
proved that minimizing the maximum interference is NP-hard whereas Bil$\grave{o}$ and
Proietti \cite{BP08} considered sender interference model and proposed a
polynomial time algorithm for minimizing the maximum interference. Their algorithm works
for many connectivity predicate like simple connectivity, strong connectivity, broadcast,
$k$-edge(vertex) connectivity, spanner, and so on. They also proved that any polynomial
time $\alpha$-approximation algorithm for minimum total range assignment problem with
connectivity predicate ${\cal P}$ can be used for designing a polynomial time
$\alpha$-approximation algorithm for minimum total interference problem for ${\cal P}$.
Panda and Shetty \cite{PS12} considered 2D networks with sender centric model and proposed an
optimal solution for minimizing the maximum interference and a 2-factor approximation algorithm
for average interference. Moscibroda and Wattenhofer \cite{MW05} proposed $O(\log n)$-factor
greedy algorithm for minimizing average interference.

\section{Minimization of Maximum Interference} \label{MinMax}
In this section we consider sender centric interference model. Given a set ${\cal V} = \{v_1,
v_2, \ldots, v_n\}$ of $n$ nodes distributed on a 2D plane, the objective is to find a range
assignment $\rho:{\cal V} \rightarrow \mathbb{R}$ such that the corresponding communication
graph $G_\rho$ contains connectivity predicate ${\cal P}$ like strong connectivity, broadcast,
$k$-edge(vertex) connectivity, spanner etc. Bil$\grave{o}$ and Proietti considered the same
problem and proposed $O({\cal P}_n \times n^2)$ time algorithm for optimum solution, where
${\cal P}_n$ is the time required to check predicate ${\cal P}$ for a given communication
graph \cite{BP08}.

Here we propose an algorithm to solve the above problem optimally. The running time of our
algorithm is $O(({\cal P}_n + n^2) \log n)$, which leads to a big improvement over \cite{BP08}.

\subsection{Algorithm}
In the network, the number of nodes is $n$, which means maximum possible interference of a node
is $n-1$. The main idea of our algorithm is very simple: first we start range assignment to each
of the node in such a way that the interference of each node is $k$ ($1 \leq k \leq n-1$). Next
we test whether the communication graph contains the connectivity predicate ${\cal P}$ or not.
If the answer is {\it yes}, then we try for lower values of $k$. Otherwise we try for higher values
of $k$. The pseudo code for the minimizing maximum sender interference (MMSI) algorithm
described in Algorithm \ref{alg-1}.

In the algorithm we use a matrix $M$ of size $n \times n-1$ and $M(i, j) =
\delta(v_i, u)$, where $u \in {\cal V}$ such that if $\rho(v_i) = \delta(v_i, u)$,
then $I_S^{v_i}(\rho) = j$ for all $i = 1, 2, \ldots n$ and $j = 1, 2, \ldots, n-1$. 
Elements of the $i$-th row are $\delta(v_i, v_1), \delta(v_i, v_2), \ldots, 
\delta(v_i, v_{i-1})$, $\delta(v_i, v_{i+1}), \ldots, \delta(v_i, v_n)$ in ascending 
order from left to right. 

\begin{algorithm}
\caption{Optimum algorithm for minimizing maximum sender interference}
\begin{algorithmic}[1]
\STATE {\bf Input:} a set ${\cal V} = \{v_1, v_2, \ldots, v_n\}$ of $n$ nodes and
                    a connectivity predicate ${\cal P}$.
\STATE {\bf Output:} a range assignment $\rho$ such that communication graph $G_\rho$ contains
                    connectivity predicate ${\cal P}$ and an interference value (MMSI).

\STATE Construct the matrix $M$ as described above.
\STATE $\ell \leftarrow 0$, $r \leftarrow n-1$
\WHILE{($\ell \neq r-1$)}
    \FOR{($i = 1, 2, \ldots, n$)}
        \STATE $\rho(v_i) = M(i, \frac{\ell+r}{2})$
    \ENDFOR
    \STATE Construct the communication graph $G_\rho$ corresponding to $\rho$.
    \STATE Test whether $G_\rho$ contains connectivity predicate ${\cal P}$ or not.
    \IF{(answer of the above test is {\it yes})}
        \STATE $r = \lfloor\frac{\ell+r}{2}\rfloor$ /* maximum interference is at most $\lfloor\frac{\ell+r}{2}\rfloor$ */
    \ELSE
        \STATE $\ell = \lfloor\frac{\ell+r}{2}\rfloor$ /* minimum interference is greater than  $\lfloor\frac{\ell+r}{2}\rfloor$ */
    \ENDIF
\ENDWHILE
\FOR{($i = 1, 2, \ldots, n$)}
    \STATE $\rho(v_i) = M(i, r)$
\ENDFOR
\STATE Return$(\rho, r)$
\end{algorithmic}
\label{alg-1}
\end{algorithm}

\begin{theorem} \label{theorem-2}
Algorithm \ref{alg-1} computes minimum of maximum sender interference (MMSI) optimally and its
worst case running time is $O(({\cal P}_n + n^2) \log n)$.
\end{theorem}

\begin{proof}
Let $\rho$ (resp. $\rho'$) be the range assignment to the nodes in ${\cal V}$ when the sender interference of
each node is $\ell$ (resp. $r$). The correctness of the Algorithm \ref{alg-1} follows from the
fact (i) the Algorithm \ref{alg-1} stops when $\ell = r - 1$ such that $G_\rho$ does not contain
connectivity predicate ${\cal P}$, whereas $G_{\rho'}$ contains connectivity predicate ${\cal P}$
and (ii) if $\rho(u) > \rho'(u)$, then $I_S^u(\rho) \geq I_S^u(\rho')$.

The construction time of matrix $M$ (line number 3 of the Algorithm \ref{alg-1}) takes
$O(n^2 \log n)$ time in worst case. Construction of the communication graph $G_\rho$
(line number 9 of the Algorithm \ref{alg-1}) and testing connectivity predicate
(line number 10 of the Algorithm \ref{alg-1}) take $O({\cal P}_n)$ time. Again, each execution
of {\bf while} loop in line number 5 reduce the value $(\ell - r)$ by half of its previous value.
Therefore, Algorithm \ref{alg-1} call this {\bf while} loop $O(\log n)$ time. Thus, the time
complexity results of the theorem follows.
\end{proof}

\section{Minimization of Total Interference} \label{MinTotal}
Given a set ${\cal V} = \{v_1, v_2, \ldots, v_n\}$ of $n$ nodes and a source node
$s \in {\cal V}$ distributed on a 2D plane, the objective is to find a range assignment
$\rho:{\cal V} \rightarrow \mathbb{R}$ such that the corresponding communication graph
$G_\rho$ contains an arborescence rooted at $s$ (connectivity predicate is broadcast)
and the total interference (sender/receiver) is minimum. Bil$\grave{o}$ and Proietti
considered the same problem and proposed $2(1 +  \ln (n-1))$-factor approximation algorithm
\cite{BP08}. Here we propose a very simple optimum algorithm. The running time of the 
algorithm is linear. The pseudo code for the minimum total sender interference (MTSI) algorithm
described in Algorithm \ref{alg-2}. Though the propose algorithm is trivial, we are proposing it 
for completeness of the literature. 

\begin{lemma} \label{lemma-1}
For any range assignment $\rho$ such that the communication graph $G_\rho = ({\cal V}, E_\rho)$
contains arborescence rooted at any node $u$ in the network of $n$ nodes, the minimum total
sender interference of the networks is at least $n-1$.
\end{lemma}

\begin{proof}
Since the communication graph $G_\rho$ contains an arborescence rooted at $u$, each
node $v \in {\cal V}\setminus \{u\}$ has an incoming edge. Therefore, the total receiver
interference is at least $n-1$. Thus, the lemma follows from Theorem \ref{theorem-1}.
\end{proof}

\begin{algorithm}
\caption{Optimum algorithm for minimizing maximum sender interference}
\begin{algorithmic}[1]
\STATE {\bf Input:} a set ${\cal V} = \{v_1, v_2, \ldots, v_n\}$ of $n$ nodes and
                    a source node $s \in {\cal V}$.
\STATE {\bf Output:} a range assignment $\rho$ such that communication graph $G_\rho$ contains
                    an arborescence rooted at $s$ and total interference (MTSI).

\FOR {($i = 1, 2, \ldots, n$)}
    \STATE $\rho(v_i) = 0$
\ENDFOR
\STATE Find the farthest node $u \in {\cal V}$ from $s$.
\STATE $\rho(s) = \delta(s, u)$
\STATE Return$(\rho, n-1)$ /* $\sum_{v \in {\cal V}}I_S^v(\rho) = n-1$ */
\end{algorithmic}
\label{alg-2}
\end{algorithm}

\begin{theorem} \label{theorem-3}
Algorithm \ref{alg-2} produces optimum total interference in $O(n)$ time.
\end{theorem}

\begin{proof}
The correctness of the algorithm follows from (i) the fact that total interference produce by
Algorithm \ref{alg-2} is $n-1$ and (ii) Lemma \ref{lemma-1}. Time complexity follows from the
{\bf for} loop (line number 3) and line number 6.
\end{proof}

\section{Heuristic for Strongly Connected Predicate} \label{Heuristic}
In this section, we propose a heuristic to minimize total interference for strongly
connected predicate i.e given a set ${\cal V} = \{v_1, v_2, \ldots, v_n\}$ of $n$ nodes
distributed on a 2D plane, we design a heuristic for range assignment $\rho$ such that
the communication graph $G_\rho$ is strongly connected. Experimental results presented
in the Subsection \ref{ExpRes} demonstrate that our heuristic perform very well compare
to existing result in the literature.

\begin{algorithm}
\caption{Heuristic to minimizing total sender interference for Strongly Connected Predicate}
\begin{algorithmic}[1]
\STATE {\bf Input:} a set ${\cal V} = \{v_1, v_2, \ldots, v_n\}$ of $n$ nodes.
\STATE {\bf Output:} a range assignment $\rho$ such that communication graph $G_\rho$ is
 strongly connected and total interference of the network.

\FOR{($i = 1, 2, \ldots, n$)}
    \STATE $\rho(v_i) \leftarrow 0$ /* initial range assignment */
    \STATE $|I_S^{v_i}(\rho(v_i))| \leftarrow 0$ /* initial interference assignment */
    \STATE $T_I \leftarrow 0$ /* initial total interference */
\ENDFOR
\STATE ${\cal U}_1 = \{v_1\}$ and ${\cal U}_2 = {\cal V} \setminus \{v_1\}$
\WHILE{(${\cal U}_2 \neq \emptyset$)}
\STATE Choose $u_1, u_1' \in {\cal U}_1$ and $u_2 \in {\cal U}_2$ such that
$|I_S^{u_1}(\delta(u_1, u_2))| - |I_S^{u_1}(\rho(u_1))| + |I_S^{u_2}(\delta(u_2, u_1'))| 
\leq |I_S^{w_1}(\delta(w_1, w_2))| - |I_S^{w_1}(\rho(w_1))| + |I_S^{w_2}(\delta(w_2, w_1'))|$
for all $w_1, w_1' \in {\cal U}_1$ and $w_2 \in {\cal U}_2$.
\STATE $T_I = T_I + |I_S^{u_1}(\delta(u_1, u_2))| - |I_S^{u_1}(\rho(u_1))| + |I_S^{u_2}(\delta(u_2, u_1'))|$
/* total interference update */
\STATE $\rho(u_1) = \delta(u_1, u_2)$ and $\rho(u_2) = \delta(u_2, u_1')$ /* new range assignments */
\STATE $|I_S^{u_1}(\rho(u_1))| = |I_S^{u_1}(\delta(u_1, u_2))|$ and
$|I_S^{u_2}(\rho(u_2))| = |I_S^{u_1}(\delta(u_2, u_1'))|$ \\
/* new interference assignments */
\STATE ${\cal U}_1 = {\cal U}_1 \cup \{u_2\}$ and ${\cal U}_2 = {\cal U}_2 \setminus \{u_2\}$
\ENDWHILE
\STATE Return($\rho, T_I$)
\end{algorithmic}
\label{alg-3}
\end{algorithm}

\begin{theorem} \label{theorem-4}
The running time of the Algorithm \ref{alg-3} is polynomial in input size.
\end{theorem}
\begin{proof}
The input size of the Algorithm \ref{alg-3} is $n$. In each execution of {\bf While} loop (line number 9 of
the Algorithm \ref{alg-3}), the size of the set ${\cal U}_2$ is decreasing by 1. Choosing the vertices
$u_1, u_1', u_2$ (line number 10 of the Algorithm \ref{alg-3}) needs $O(n^3)$ time. Thus, the time
complexity result of the theorem follows.
\end{proof}

\subsection{Experimental Results} \label{ExpRes}
The model consists of $n$ sensor nodes randomly distributed in a $1000 \times 1000$ square grid. 
For different values of $n$, we execute our heuristic 100 times for different input instances. 
We have taken average value of the total interference of the nodes in the network. In Figure 
\ref{figure-1}, we have shown average interference of our heuristic and compare it with the 
algorithm improved SMIT \cite{PS12}. In Table \ref{table-1}, we have shown the comparison of 
the total interference between our heuristic and the algorithm proposed in \cite{PS12}. These 
experimental results demonstrate that our heuristic outperform existing algorithm.

\begin{figure}[h]
\centering
\centerline{\psfig {figure=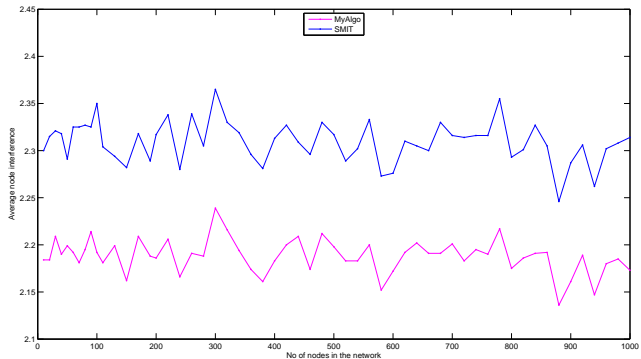,width=4in}}
\caption{Average node interference}
\label{figure-1}
\end{figure}

\vspace{-0.5in}
\begin{table}[h]
\begin{center} \label{table-1}
\begin{tabular}{|c|c|c|}
\hline
\# of & \multicolumn{2}{|c|}{Total Average Interference} \\
nodes $(n)$ & Our Heuristic & Improved SMIT \cite{PS12}\\
\hline
10 & 2.184 & 2.30 \\
\hline
20 & 2.184 & 2.315 \\
\hline
30 & 2.209 & 2.321 \\
\hline
50 & 2.199 & 2.291 \\
\hline
70 & 2.181 & 2.325 \\
\hline
100 & 2.192 & 2.350 \\
\hline
200 & 2.186 & 2.317 \\
\hline
300 & 2.239 & 2.365 \\
\hline
400 & 2.183 & 2.313 \\
\hline
500 & 2.198 & 2.317 \\
\hline
700 & 2.201 & 2.316 \\
\hline
900 & 2.161 & 2.287 \\
\hline
1000 & 2.173 & 2.314 \\
\hline
\end{tabular}
\caption{Simulation results for total interference}
\end{center}
\end{table}

\section{Conclusion} \label{Conclusion}
In this paper we considered 2D networks i.e., the wireless nodes are distributed on a
2D plane. All the results presented in this paper are applicable in 3D networks also.
Here we considered interference (sender interference model) minimization problem
in wireless networks. We proposed algorithm for {\it minimizing maximum interference} and
{\it minimizing total interference} of a given networks. For minimizing maximum interference we
presented optimum solution with running time $O(({\cal P}_n + n^2) \log n)$ for connectivity
predicate ${\cal P}$ like strong connectivity, broadcast, $k$-edge(vertex) connectivity, spanner,
where $n$ is the number of nodes in the network and $O({\cal P}_n)$ is the time complexity for
checking the connectivity predicate ${\cal P}$. The running time of the previous best known
solution was $O({\cal P}_n \times n^2)$ \cite{BP08}. Therefore, our solution is a significant
improvement over the best known solution with respect to time complexity.

For the minimizing total interference we proposed optimum algorithm for connectivity predicate 
broadcast. The running time of the proposed algorithm is $O(n)$. For the same problem, the previous
best known result was $2(1 + \ln (n-1))$-factor approximation algorithm \cite{BP08}. Therefore,
our solution is a significant improvement over the existing solution in the literature. We also
proposed a heuristic for minimizing total interference in the case of strongly connected predicate
and compare our result with the best result available in the literature. Experimental results
demonstrate that our heuristic outperform existing result.

\end{document}